\documentclass[11pt]{article}%{llncs}

\usepackage{fullpage}
\usepackage{amssymb,amsmath}
\usepackage{graphicx, epsfig}

\def\id#1{\ensuremath{\mathit{#1}}}

\def\idrm#1{\ensuremath{\mathrm{#1}}}
\def\idtt#1{\ensuremath{\mathtt{#1}}}

\def\ceil#1{\lceil #1 \rceil}

\newtheorem{theorem}{Theorem}
\newtheorem{lemma}{Lemma}
\newtheorem{fact}{Fact}

\newenvironment{proof}{\trivlist\item[]\emph{Proof}:}%
{\unskip\nobreak\hskip 1em plus 1fil\nobreak$\Box$
\parfillskip=0pt%
\endtrivlist}

\newenvironment{itemize*}%
  {\begin{itemize}%
    \setlength{\itemsep}{0pt}%
    \setlength{\parskip}{0pt}%
    \setlength{\parsep}{0pt}%
    \setlength{\topsep}{0pt}%
    \setlength{\partopsep}{0pt}%
  }%
  {\end{itemize}}%

\newcommand{\halfleftsect}[2]{(#1,#2]}
\newcommand{\halfrightsect}[2]{[#1,#2)}
\newcommand{\cD}{{\cal D}}
\newcommand{\cI}{{\cal I}}
\newcommand{\cV}{{\cal V}}

\newcommand{\oD}{\overline{D}}
\newcommand{\oI}{\overline{I}}

\newcommand{\tL}{\widetilde{L}}
\newcommand{\tE}{\widetilde{E}}

\newcommand{\parent}{\idrm{parent}}
\newcommand{\subs}{\idtt{subset}}
\newcommand{\desc}{\idtt{desc}}
\newcommand{\polylog}{\idrm{polylog}}
\newcommand{\ilog}{\idrm{ilog}}
\newcommand{\anc}{\idrm{anc}}
\newcommand{\ptr}{\id{ptr}}
\newcommand{\eps}{\varepsilon}

\begin{document}

%% \title{\vspace*{3cm} }
%% \author{Yakov Nekrich\thanks{Department of Computer Science, 
%% University of Bonn. 
%% Email {\tt yasha@cs.uni-bonn.de}}
%% }
%% \date{\today}
%% \maketitle
%% \begin{abstract}

%% \end{abstract}
%% \thispagestyle{empty}
%% \clearpage

\title{External Memory Orthogonal Range Reporting with Fast Updates}
\author{
Yakov Nekrich\thanks{
Department of Computer Science, 
%University of Bonn. 
University of Chile.
Email {\tt yakov.nekrich@googlemail.com}}
}
%\institute{}
\date{} %\today}
 \maketitle
\begin{abstract}
In this paper we describe data structures for orthogonal range reporting 
in external memory that support fast update operations. The query costs 
either match the query costs of the best previously known data structures
or differ by a small multiplicative factor.
\end{abstract}
\section{Introduction}
In the orthogonal range reporting problem a set of points is stored in 
a data structure so that for any $d$-dimensional query range 
$Q=[a_1,b_1]\times\ldots\times [a_d,b_d]$ all points that belong to $Q$
 can be reported.  Due to its fundamental nature and its applications,
the orthogonal range reporting problem was studied extensively;  
we refer to e.g.~\cite{GBT84,Ch88,VV96,ABR00,C11} for a small selection 
of important  
publications.  In this paper we address the issue of constructing  
dynamic data structures that support fast update operations in the 
external memory model. 

%The external memory model DESCRIBE THE MODEL HERE. 

\tolerance=1000
External memory data structures for orthogonal range reporting also received 
significant
 attention, see e.g.,~\cite{RS94,SR95,VV96,ASV99,A08,N08,AAL09,N10}.  
We refer to~\cite{V01} for the definition of the external memory model and 
a survey of previous results.
In particular, dynamic 
data structures for $d=2$ dimensions are described in~\cite{RS94,SR95,ASV99}.  
The best previously known data structure of Arge, Samoladas, 
and Vitter~\cite{ASV99} uses $O((N/B)\log_2 N/\log_2\log_B N)$ blocks of 
space, 
answers queries in $O(\log_B N+K/B)$ I/Os and supports updates 
in $O(\log_B N (\log_2 N/\log_2\log_B N))$ I/Os; in~\cite{ASV99}, the authors 
also show that the space usage of their data structure is optimal.  
Recently, the first dynamic data structure that supports queries 
in $O(\log^2_B N+K/B)$ I/Os in  $d=3$ dimensions was described~\cite{N10}.

All previously described external memory data structures with 
optimal or almost-optimal query cost need 
$\Omega((\log_B N\log_2 N)/\log_2\log_B N)$ I/Os to support an 
insertion or a deletion of a point; see Table~\ref{tbl:results}. 
This  compares unfavorably with significantly lower update costs that can 
be achieved 
by  internal memory data structures. For instance, the two-dimensional 
data structure 
of Mortensen~\cite{M06} supports updates in $O(\log_2^{f}N)$ time for any
 constant $f> 7/8$. 
Moreover, the update costs of previously described external structures 
 contain an $O(\log_2 N)$ factor. Since block size $B$ can be large, achieving 
update cost that only depends on $\log_B N$ would be desirable.   
High  cost of updates is also a drawback of the three-dimensional 
data structure described in~\cite{N10}.  Reducing the cost of update 
operations can be important in the dynamic scenario when the data 
structure must be updated frequently. 

{\bf Our Results.}
We describe several data structures for orthogonal range reporting queries 
in $d=2$ dimensions that achieve lower update costs. 
We describe two data structures that 
 support queries in $O(\log_B N + K/B)$ I/Os. 
These data structures 
support updates in $O(\log_B^{1+\eps}N)$ I/Os  
with high probability 
and in $O(\log_B^2 N)$ deterministic I/Os respectively.  
Henceforth $\eps$ denotes an arbitrarily small positive constant. 
We also describe a data structure that uses $O((N/B)\log_2 N)$ blocks of 
space, answers queries in $O(\log_B N(\log_2\log_B N)^2)$ I/Os, 
and supports updates in $O(\log_BN(\log_2\log_B N)^2)$ I/Os. 
%For $d=3$, we improve the update cost of the data structure 
%described in~\cite{N10} to $O(\log_B^{2+\eps}N)$ I/Os. 
All our results are listed in Table~\ref{tbl:results}.

\begin{table}[tbp]
\centering
\resizebox{.98\textwidth}{!}
{\begin{tabular}[h]{|c|c|c|c|}
\hline
 Source        & Query  & Update & Space\\ 
               & Cost   &  Cost  & Usage  \\ \hline 
% $d=2$: &              &       &          \\
 \cite{RS94}  &  $O(\log_B N + \frac{k}{B})$       &  $O(\log_B N \log_2 N \log_2^2 B)$    & $O((N/B)\log_2 N \log_2 B \log_2\log_2 B)$    \\
\cite{SR95}  &  $O(\log_B N + \frac{k}{B}+\ilog(B))$     & $O(\log_2 N(\log_B N + (\log_B^2 N)/B)) $  &  $O((N/B)\log_2 N)$ \\
\cite{ASV99} & $O(\log_B N + \frac{k}{B})$  & $O(\log_B N \log_2 N /\log_2 \log_B N)$ & $O((N/B)\log_2 N/\log_2 \log_B N)$ \\
      *       & $O(\log_B N + \frac{k}{B})$  & $O(\log^{1+\eps}_B N )\,\dagger$  & $O((N/B)\log_2 N)$  \\
      *       & $O(\log_B N + \frac{k}{B})$   & $O(\log^{2}_B N )$ & $O((N/B)\log_2 N/\log_2 \log_B N)$ \\
      *       & $O(\log_B N(\log_2 \log_B N)^2 + \frac{k}{B})$   & $O(\log_B N(\log_2 \log_B N)^2 )$ & $O((N/B)\log_2 N)$ \\
% \hline
%   $d=3$ &              &       &           \\
%  \cite{N10}  &  $O(\log^2_B N+ \frac{k}{B})$  & $O(\log_2^3 N)$        &     $O((N/B)\log_2^2 N \log^2_2 B)$\\ 
%      *       &  $O(\log^2_B N+ \frac{k}{B})$  & $O(\log_B^{2+\eps} N)\,\dagger$  & $O((N/B)\log_2^2 N \log^2_2 B)$\\ 
\hline
\end{tabular}
}
\caption{New data structures and some previous results 
for $d=2$  dimensions. 
Our results are  marked with an asterisk; $\dagger$ denotes randomized 
results. The result in the first row of the table can be obtained from the
 result in~\cite{RS94} using a standard technique.
The function $\ilog(x)$ is the iterated $\log^*$ function: $\ilog(x)$  
denotes the number of times we must apply the $\log^*$ function to $x$ before 
the result becomes $\leq 2$, where 
$\log^*(x)=\min\{t\,|\,\log_2^{(t)}(x)<2\,\}$, 
and $\log_2^{(t)}(x)$ denotes the $\log_2$ function repeated $t$ times.
}
\label{tbl:results}
\end{table}
{\bf Overview.}
The situations when the block size $B$ is small and when $B$ is not so small 
are handled separately.   
If the block size is sufficiently large, $B=\Omega(\log_2^4 N)$ for an 
appropriate choice of constant, 
our construction is based  on the bufferization technique.
We show that a batch of $O(B^{1/4})$ queries can be processed 
with $O(\log_B N)$ I/Os. Hence, we can achieve constant amortized update 
cost for sufficiently large $B$.
In the case when $B$ is small, $B=O(\log_2^4 N)$, we construct the 
base tree with fan-out $\log_2^{\eps}N$ or the base tree with constant fan-out. 
Since $B=\polylog_2(N)$, the height of the base tree is bounded by 
$O(\log_B N)$ or $O(\log_B N\log_2\log_B N)$. 
Hence, we can reduce a two-dimensional query to a small number of simpler
queries.

In section~\ref{sec:3sid} we describe a data structure that supports 
three-sided reporting queries in $O(\log_B N+\frac{K}{B})$ I/Os and updates 
in $O(\frac{1}{B^{\delta}})$ I/Os if $B=\Omega(\log_B^4 N)$. 
Henceforth $\delta$ denotes an arbitrary positive constant, such that 
$\delta\leq 1/4$.
%In section~\ref{sec:2dim}, 
In Appendix A, we generalize this result and obtain a data structure 
that supports updates in $O(1)$ I/Os and orthogonal range reporting 
queries in $O(\log_B N+\frac{K}{B})$ I/Os if $B=\Omega(\log^4_2 N)$. 
Thus if a block size is sufficiently large, 
 there exists a data structure with optimal query cost and $O(1)$ amortized 
update cost. We believe that this result is of independent interest.  
Data structures for $B =O(\log^4_2 N)$ are described in
 section~\ref{sec:smallb}.

%\section{Preliminaries}

\section{Three-Sided Range Reporting for $B=\Omega(\log_B^4 N)$}
\label{sec:3sid}
Three-sided queries are a special case of two-dimensional 
orthogonal range queries. The range of a three-sided query 
is the product of a closed interval and a half-open interval. 
In this section we assume that the block size $B\geq 4h\log_B^4 N$
for a constant $h$ that will be defined later in this section.
Our data structure answers three-sided queries with $O(\log_B N + K/B)$ 
I/Os  and updates are supported in $O(1/B^{\delta})$ amortized I/Os. 

Our approach is based on a combination of external priority 
tree~\cite{ASV99} with buffering technique~\cite{A03}. Buffering was 
previously used to answer searching and reporting problems in one dimension. 
In this section we show that buffering can be applied to three-sided range 
reporting problem in the case when $B=\Omega(\log_B^4 N)$. 
At the beginning, we describe the external priority tree~\cite{ASV99} data 
structure. 
Then, we show how this data structure 
can be modified so that a batch of $B^{\delta}$ updates can be processed 
in constant amortized time.  Finally, we describe the procedure for reporting 
all points in a three-sided range $Q=[a,b]\times\halfrightsect{c}{+\infty}$. 
%To simplify the description, we ignore the issues related to rounding the 
%number of points. 

The following Lemma is important for our construction.
\begin{lemma}\label{lemma:small3sid}
A set $S$ of $O(B^{1+\delta})$ points can be stored in a data structure that 
supports three-sided reporting queries in $O(K/B)$ I/Os, where 
$K$ is the number of points in the answer; this data structure can 
be constructed with $O(B^{\delta})$ I/Os.
\end{lemma}
\begin{proof}
We can use the data structure of Lemma 1 from~\cite{ASV99}.
\end{proof}

{\bf External Priority Tree.}
%{\bf Static Data Structure.}
Leaves of the external priority tree contain the  $x$-coordinates of points 
in sorted order. 
Every leaf contains $\Theta(B)$ points and each internal node 
has $\Theta(B^{\delta})$ children. 
We assume throughout this section that the height of an external priority 
tree is bounded by $h\log_B N$.  
The \emph{range} $rng(v)$ of a node $v$ is the interval bounded by 
the minimal and the maximal 
coordinates stored in its leaves; we say that a point $p$ belongs to 
(the range of) a node $v$ if its $x$-coordinate belongs to the range of $v$.
Each node is associated with a set $S(v)$, $|S(v)|=\Theta(B)$, defined as 
follows. 
Let $L(v)$ denote the set of all points that belong to the range of $v$. 
The set $S(v)$ contains $B$ points with largest $y$-coordinates 
among all points in $L(v)$ that do not belong to any set $S(w)$, where 
$w$ is an ancestor of $v$. Thus external priority tree is a modification 
of the priority tree 
with node degree $B^{O(1)}$, such that each node contains $\Theta(B)$ points. 
  
The data structure $F(v)$ contains points from $\cup S({v_i})$ for all 
children $v_i$ of $v$. By Lemma~\ref{lemma:small3sid}, $F(v)$ 
supports three-sided queries in $O(1)$ I/O operations. 
Using $F(v)$, we can answer three-sided queries in $O(\log_B N + K/B)$ 
I/Os; the  search procedure is described in~\cite{ASV99}.  \\

{\bf Supporting Insertions and  Deletions.} 
Now we describe a data structure that supports both insertions and deletions. 
We will show below how a batch of inserted or deleted points can be processed 
efficiently. The main idea is to maintain buffers with inserted and deleted 
points in all internal nodes. The buffer $D(v)$, $v\in T$, contains  
points that are stored  in descendants of $v$ and must be deleted. The buffer 
$I(v)$, $v\in T$, contains points that must be inserted into sets $S(u)$ for 
a descendant $u$ of $v$. 
A buffer can contain up to $B^{3\delta}$ elements. 
When a buffer $I(v)$ or $D(v)$ is full, we flush it into the 
children $v_j$ of $v$; all sets $I(v_j)$, $D(v_j)$, and $S(v_j)$ are updated
accordingly.

Definitions of $S(v)$ and $F(v)$ are slightly modified for the dynamic
 structure. 
Every set $S(v)$ contains at most $2B$ points. If $S(v)$ contains less than 
$B/2$ points than $S(v_i)=\emptyset$ for each child $v_i$ of $v$.
The data structure $F(v)$  contains all points from $S(v_j)\cup I(v_j)$ for 
all children $v_j$ of $v$. 
We store an additional data structure $R(v)$ in each internal node. 
$R(v)$ contains all points from $\cup S(v_i)$ for all children $v_i$ of $v$. 
$R(v)$ can be constructed in $O(B^{\delta})$ I/Os; we can obtain $B^{3\delta}$ 
points with highest $y$-coordinates stored in $R(v)$ in $O(1)$ I/Os.
Implementation of $R(v)$ is very similar to implementation of $F(v)$; details 
will be given in the full version.

Suppose that all points from the set   $\cD$, $|\cD|=O(B^{\delta})$, 
 must be deleted.
We remove all points from $\cD\cap S(v_r)$ and $\cD\cap I(v_r)$ 
from $S(v_r)$ and $I(v_r)$ respectively. We set $D(v_r)=D(v_r)\cup 
(\cD\setminus S(v_r))$. When $D(v)$ for an internal node 
$v$ is full, $|D(v)|=B^{3\delta}$,  we distribute the points of 
$D(v)$ among the children $v_j$ of $v$. Let $D_j(v)$ be the points of 
$D(v)$ that belong to the range of $v_j$ and update $S(v_j)$, $D(v_j)$ 
as described above:
We remove all points from $D_j(v)\cap S(v_j)$ and $D_j(v)\cap I(v_j)$ 
from $S(v_j)$ and $I(v_j)$ respectively. All  points of 
$D_j(v)\setminus S(v_j)$ are inserted into $D(v_j)$. 
Finally, we update $F(v)$ and $R(v)$. 

%If the number of points in $S(v_j)$ is smaller than $B/2$ when $D(v)$ is 
%emptied, we move the points from children of $v$ into $S(v)$. 

We can insert a batch $\cI$ of $B^{\delta}$ points using a similar procedure. 
Initially, all points from $\cI$ are inserted into 
buffer $I(v_r)$ or $S(v_r)$ and points of $\cI \cap D({v_r})$ are removed 
from $D({v_r})$. Let $S'(v_r)=S(v_r)\cup \cI$ and let 
$S''(v_r)$ be the set of $B$ points with highest $y$-coordinates in $S'(v_r)$. 
We set $S(v_r)=S''(v_r)$ and $I(v_r)=I(v_r)\cup (S'(v_r)\setminus S''(v_r))$. 
When the buffer $I(v)$ in an internal  node $v$ is full, 
$|I(v)|\geq B^{3\delta}$, we update the sets $S(v_j)$ and $I(v_j)$ in the 
children $v_j$ of $v$. Let $I_j(v)$ be the set of points in $I(v)$ 
that belong to the range of $v_j$. Let $S'(v_j)=S(v_j)\cup I_j(v)$
and let $S''(v_j)$ be the set of $B$ points with the highest $y$-coordinates 
in $S'(v_j)$. We set $S(v_j)=S''(v_j)$, 
$D(v_j)=D(v_j)\setminus (D(v_j)\cap I_j(v))$, and 
$I(v_j)=I(v_j)\cup (S'(v_j)\setminus S''(v_j))$. 
The data structures $F(v)$ and $R(v)$ are updated accordingly.

When a buffer $I(v)$ is full, we can re-build all $I(v_j)$, $D(v_j)$, $S(v_j)$ 
and the data structures $F(v)$, $R(v)$  in $O(B^{\delta})$ I/Os. 
Each inserted point is inserted in $O(\log_B N)$ buffers $I(v)$. 
Hence, an amortized cost of re-building secondary data structures 
caused by an insertion is $O(B^{\delta}\log_B N/B^{3\delta})=O(1/B^{\delta})$. 
The cost of a deletion can be analyzed in the same way.

We  also take care that the number of points stored in 
sets $S(u)$ is not too small. 
Suppose that the number of points in some $S(w)$ is smaller than $B/2$ when 
$D(\parent(w))$ is emptied. If $w$ is a leaf or $S(w_j)=\emptyset$ for 
all children $w_j$, we do not need to rebuild $S(w)$.
Otherwise,  we move some  points from $S(w_j)$  into $S(w)$.  
Using the data structure $R(w)$, we identify 
$B-|S(w)|$ points with the highest $y$-coordinates in 
$\cup_j S(w_j)$.  These points are removed from $R(w)$, $F(w)$, $S(w_j)$ 
and inserted into $S(w)$. We also update $F(\parent(w))$ and 
$R(\parent(w))$.  For every child $w_j$ of $w$, we recursively call the 
same procedure. The total cost of updating all data structures in a node  is 
$O(B^{1-3\delta})$. Using standard analysis, we can show that maintaining 
the size of $S(w)$ incurs an amortized cost $O(1/B^{\delta})$.

Besides that, we should  take care that each leaf contains $x$-coordinates 
of at most $2B$ points. To maintain this invariant, 
the external priority  tree is implemented as a WBB-tree~\cite{AV03}. 
The branching parameter of our  WBB-tree equals  to $B^{\delta}$ and the 
leaf parameter equals to $B$.
%The height of the tree does not exceed $h\log_B N$ for a constant $h$. 
When the total number of points stored in all descendants of 
a node $u$ equals to  $2B^{\ell\delta}\cdot B$, 
we split the node $u$ into $u'$ and $u''$. 
A node on level $\ell$ is split at most once after a series of 
$\Theta(B^{\ell\delta}\cdot B)$ insertions. 
When a node is split, we assign each element of $S(u)$, $I(u)$, and 
$D(u)$ to the corresponding set in $u'$ or $u''$. 
As a result, either $S(u')$ or $S(u'')$ may contain less than $B/2$
elements. In this case, we move the points from descendants of 
$u'$ into $u'$ (from descendants of $u''$ into $u''$) as described 
above. The total amortized cost of  splitting a node is $O(1/B^{\delta})$.

{\bf Answering Queries.} 
Consider a query $Q=[a,b]\times \halfrightsect{c}{+\infty}$. 
Let $\pi$ denote  the set of all nodes that lie 
on the path from the root to $l_a$ or on the path from the root 
to $l_b$, where $l_a$ and $l_b$ are the leaves that contain $a$ and $b$ 
respectively. Then all points inside the range 
 $Q$ 
are stored in sets 
$S(v)$ or $I(v)$, where the node $v$ belongs to $\pi$ or $v$ is a descendant 
of a node that belongs to $\pi$. 
Two following facts play crucial role in the reporting procedure. 
\begin{fact}\label{fact:ord}
Let $w$ be an ancestor of a node $v$.
For any $p\in S(v)$ and $p'\in S(w)$, $p.y < p'.y$. 
For any $p\in I(v)$ and $p'\in S(w)$, $p.y < p'.y$. 
\end{fact}
\begin{fact}\label{fact:del}
Suppose that a point $p\in S(v)$ is deleted from $S$ (but $p$ is not deleted 
from $S(v)$ yet). Then, $p$ belongs to a set $D(w)$ for an ancestor
 $w$ of $v$.
\end{fact}
We set the value of the constant $h$ so that  the height of $T$ 
does not exceed $h\log_B N$. 
As follows from the Fact~\ref{fact:del}, 
 the total number of deleted  points in $S(v)$ is bounded by 
$h\cdot B^{3\delta}\log_B N\leq B/4$. 
%Let $DEL(v)=\cup D(w)$ where $w=v$ or  $w$ is an ancestor of $v$. 
Let $DEL(v)=\cup_{w=\anc(v)}(D(w)\setminus \cup_{w'=\anc(w)} I(w'))$, where 
$\anc(u)$ denotes an ancestor of a node $u$. To wit,  
$DEL(v)$ is the set of all points $p$, such that $p$ belongs to 
some set $D(w)$ for an ancestor  $w$ of $v$, but $p$ does not belong 
to any $I(w')$ for an ancestor $w'$ of $w$. 
By 
Fact~\ref{fact:del} all points in $S(v)\cup I(v)$ 
that are already deleted from the data structure belong to $DEL(v)$.
If the set $DEL(v)$ is known, then $DEL(v_i)$ for a child $v_i$ of $v$ 
can be constructed in $O(1)$ I/Os. Therefore we can construct  $DEL(v)$ 
for all $v\in \pi$ in $O(\log_B N)$ I/Os. 

We can output all points that belong to $Q$ using the following 
procedure. Let $\pi_1$ be the path from $l_a$ to the lowest common ancestor
 $v_l$ of $l_a$ and $l_b$. Let $\pi_2$ be the path from $l_b$ to $v_l$.
First, we examine all nodes $v\in \pi$ and report all points 
$p\in (S(v)\cup I(v)) \setminus DEL(v)$ 
that belong to $Q$; this can be done with $O(\log_B N)$ I/Os. 
All other points in $S\cap Q$ are stored in a set $S_u$ where 
$u$ is a descendant of some $v\in \pi_1\cup \pi_2$ or $u$ is a 
descendant of $v_l$. 

Consider a node $v\in \pi_2$, such that $v\not=v_l$; we will show how 
points in $S(u)\cap Q$ for all descendants $u$ of $v$ can be reported.
Suppose that the child $v_i$ of 
$v$ also belongs to $\pi_2$ and $rng(v_{i-1})=[a',b']$. 
Let $Q_v=[a,b']\times\halfrightsect{c}{+\infty}$.  For a point $p$ stored in a 
descendant $u$ of $v$ such that $u\not\in \pi_2$, $p$ belongs to $Q$ if and 
only if 
$p$ belongs to $Q_v$. All $p\in Q_v\cap S(u)$ are reported as follows. 
Initially we set $u=v$. We identify all points stored in 
$S(u_i)\cap Q_v$ or $I(u_i)\cap Q_v$ for some child $u_i$ of $u$ 
using  the data structure $F(u)$. 
Then, we process the resulting list of points and remove all 
points that belong to $DEL(u)$. 
Finally, we identify all non-leaf children $u_i$ of $u$ such that
at least $B/2$ points from $S(u_i)$ are reported. 
%the point $p_i \in G(u)\cap S(u_i)$ belongs to $Q_v$.  
%$S(u_i)$ contains at least $B/2$ points with $p.y \geq c$. 
We visit every such $u_i$, compute $DEL(u_i)$, and recursively 
call the same procedure in $u_i$.

Our procedure reports all points in $L(v)\cap Q_v$. 
Suppose that we visited a node $u$, but the child $u_j$ of $u$ 
was not visited. All points from $(S(u_j)\cup I(u_j))\cap Q$ were 
reported when the node $u$ was visited. Since $S(u_j)$ contains at 
least $B/2$ points, at least one point $p_j\in S(u_j)$ does not belong 
to $Q$. The $x$-coordinate of $p_j$ belongs to $[a,b]$; hence, the 
$y$-coordinate of $p_j$ is smaller than $c$. By Fact~\ref{fact:ord}, 
$y$-coordinates of all points stored in $S(\nu)\cup I(\nu)$ for any  
descendant $\nu$ of $u_j$ are smaller than $c$. Hence, all points 
$p\in S(\nu)\cup I(\nu)$ are not relevant for our query.

\tolerance=1400
The search procedure spends $O(1)$ I/Os in every visited node 
(ignoring the cost of reporting points). Let $K_v$ be the total number 
of reported points in $L(v)\cap Q_v$. A node $u$ is visited if at least 
$B/4$ points from $S(u)$ were reported.
Thus we can charge at least $B/4$ points for every 
visited node. 
We can conclude that the search procedure spends $O(K_v/B)$ I/Os in 
the descendants of $v$. 
 Descendants of nodes $v\in \pi_1$, $v\not=v_l$, 
and descendants of $v_l$ can be 
processed with a similar procedure. Therefore the total query cost 
is $O(\log_B N + K/B)$. 
We obtain the following result.
\begin{lemma}\label{lemma:3sid}
Suppose that $B^{\delta}\geq 4h\log_B N$ for a constant $h$ defined above 
and some $\delta\leq 1/4$.
Then there exists a data structure that uses $O(N/B)$ blocks of space and 
answers three-sided reporting 
queries in $O(\log_B N + K/B)$ I/Os. 
The amortized cost of inserting or deleting a  point 
is $O(1/B^{\delta})$.
\end{lemma}
A  similar  approach can be used 
to construct the data structure for  general two-dimensional range reporting
 queries. 
\begin{lemma}\label{lemma:imp2dim}
Suppose that $B^{\delta}\geq 4h_1\log_2 N$ for a constant $h_1$ defined in 
Appendix A 
%defined above 
and some  $\delta\leq 1/4$.
Then there is a data structure that uses $O((N/B)\log_2 N/\log_2\log_B N)$ 
blocks of space and answers two-dimensional orthogonal range  reporting 
queries in $O(\log_B N + K/B)$ I/Os. 
The amortized cost of inserting or deleting a  point is $O(1)$.
\end{lemma}
Our data structure  uses the bufferization technique, 
but some additional ideas 
are also needed to retain the $O(\log_B N + K/B)$ query cost and 
achieve the optimal space usage.
We provide the details in the
 Appendix A.

%\section{Two-Dimensional Range Reporting for $B=\Omega(\log_2^4 N)$}
%\label{sec:2dim}
%In this section we show how a similar batching approach can be used 
%to construct the data structure for  general two-dimensional range reporting
% queries. We consider the situation when $B^{\delta}\geq 4h_1\log_2 N$ for a 
%constant $h_1$ that will be defined below.
%
%We use the standard approach of reducing a general range reporting query 
%to two three-sided reporting queries. As in section~\ref{sec:3sid}, we also 
%apply the  bufferization technique. 
%
%
%\begin{lemma}\label{lemma:imp2dim}
%Suppose that $B^{\delta}\geq 4h_1\log_2 N$ for a constant $h_1$ described 
%above and any $\delta\leq 1/4$.
%Then there exists a data structure that uses $O((N/B)\log_2 N/\log_2\log_B N)$ 
%blocks of space and answers two-dimensional orthogonal range  reporting 
%queries in $O(\log_B N + K/B)$ I/Os. 
%The amortized cost of inserting or deleting a  point is $O(1)$.
%\end{lemma}

\section{Two-Dimensional Range Reporting for small $B$}
\label{sec:smallb}
It remains to consider the case when the block size $B$ is small. 
In this section we assume that 
$B=O(\log_2^4 N)$ and describe several data structures for this case.

{\bf Reduction to Three-Sided Queries.}
We use the base structure that is similar to structures in~\cite{SR95,ASV99}.
We construct the base tree $T$ with fan-out $\rho=\Theta(\log_2^{\eps} N)$ 
on the set of $x$-coordinates. 
In every node $v$ of $T$ we store the data structures that 
support three-sided queries  $\halfrightsect{a}{+\infty}\times [c,d]$ and 
$\halfleftsect{-\infty}{b}\times [c,d]$ .  The data structures for three-sided 
queries are implemented using the external priority search tree~\cite{ASV99}, 
so that the query and update costs are   $O(\log_B N+K/B)$ and $O(\log_B N)$.  
In every node $v$, we also store a data structure that supports the 
following queries: for any $c<d$ and for any $1\leq i\leq j\leq \rho$, 
we can report all points $p$, such that $p.y \in [c,d]$ and 
$p$ is stored in the child $v_f$ of $v$, $i\leq f\leq j$. 
In~\cite{ASV99} the authors describe a  linear space 
data structure that supports  such queries in $O(\rho +K/B)$ I/Os and 
updates in $O(\log_B N)$ I/Os. 

To answer a query $[a,b]\times [c,d]$, we identify the lowest node $v$, 
such that $[a,b]\subset rng(v)$. Suppose that $[a,b]$ intersects with 
$rng(v_l),rng(v_{l+1}),\ldots rng(v_r)$. We answer three-sided 
queries $\halfrightsect{a}{+\infty} \times [c,d]$
and $\halfleftsect{-\infty}{b}\times [c,d]$ on data structures for nodes 
$v_l$ and $v_r$ respectively. Then, we report all points $p$ with 
$c\leq p.y\leq d$ stored in the nodes $v_{l+1},\ldots, v_{r-1}$. Since 
$\rho=O(\log^{\eps}N)=O(\log_B N)$, the total query cost is 
$O(\log_B N + \frac{K}{B})$.

Since each point is stored in $O(\log_2 N/\log_2 \log_2 N)$ data structures, 
the space usage of structure is $O((N/B)\log_2 N/\log_2\log_2 N)=
O((N/B)\log_2 N/\log_2\log_B N)$ because $B=O(\log_2^4 N)$. 
The update cost is $O(\log_B N(\log_2 N/\log_2 \log_2 N))=O(\log_B^2 N)$. 
Combining this result with Lemma~\ref{lemma:imp2dim}, we obtain the following 
Theorem
\begin{theorem}
\label{theor:2dim1}
There is a data structure that uses $O((N/B)\log_2 N/\log_2\log_B N)$ 
blocks of space and answers orthogonal range reporting queries in two 
dimensions 
in $O(\log_B N+ \frac{K}{B})$ I/O operations. Updates are supported in 
$O(\log_B^{2} N)$ amortized I/Os.
\end{theorem}

{\bf Reduction to One-Dimensional Queries.}
We can obtain further results by reducing a two-dimensional query 
to a number of one-dimensional queries.
We construct a standard range tree with constant fan-out on the $x$-coordinates 
of points. 
All points that belong to a range of a node $v$ are stored in $v$.
For any interval $[a,b]$, we can find $O(\log_2 N)$ nodes $u^i$, such 
that  $p.x\in [a,b]$ if and only if $p$ is stored in a node $u^i$. 
Hence, all points in the query range $Q=[a,b]\times [c,d]$ can be reported 
by answering a one-dimensional query $Q_y=[c,d]$ in 
$O(\log_2 N)$ nodes $u^1,\ldots,u^t$ of the range tree.
Using the fractional cascading technique, 
we can find the predecessor $d(u^i)$ of $d$ and the successor $c(u^i)$ 
of $c$ in all nodes $u^i$ in 
$O(\log_2 N\log_2\log_2 N)$ time; we refer to e.g.,~\cite{MN90} 
for details.  When we know $c(u^t)$ and $d(u^t)$ we can report 
all elements stored in the node $u^t$ in  $O(K/B)$ I/Os. 
Hence, the query cost is 
$O(\log_2 N\log_2\log_2 N+K/B)=O(\log_ B N (\log_2 \log_B N)^2 +K/B)$. 
Each point is stored in $O(\log_2 N)$ secondary data structures. 
As described in~\cite{MN90}, the range tree augmented with fractional cascading 
data structures can be updated in  $O(\log_2 N (\log_2\log_2N))=
O(\log_BN(\log_2\log_B N)^2)$ time; hence, an update requires 
$O(\log_BN(\log_2\log_B N)^2)$ I/O operations. 
\begin{theorem}
\label{theor:2dim2}
There exists a data structure that uses $O((N/B)\log_2 N)$ blocks of 
space and answers orthogonal range reporting queries in two dimensions 
in $O(\log_B N(\log_2 \log_B N)^2 + \frac{K}{B})$ I/O operations. 
Updates are supported in $O(\log_B N(\log_2\log_B N)^2)$ amortized I/Os.
\end{theorem}

{\bf Range Trees with $B^{O(1)}$ Fan-Out.}
Let $\eps'=\eps/10$. 
If points have integer coordinates, we can reduce the query cost by 
constructing  a range tree with 
fan-out $B^{\eps'}$. 
For every node $v$ and every pair of indexes $i\leq j$, where $v_i$, $v_j$ 
are the children of $v$,  all 
points that belong to the children $v_i,\ldots, v_j$ of $v$ belong 
to a list $L_{ij}(v)$.  
A data structure $E_{ij}(v)$ supports \emph{one-dimensional 
one-reporting} queries on a set of integers. 
That is,  $E_{ij}(v)$  enables us  to find for any 
interval $[c,d]$ some point $p\in L_{ij}(v)$ such that $p.y\in [c,d]$,
if such $p$ exists and if all points have integer  coordinates. 
As described in~\cite{MPP05}, we can implement $E_{ij}(v)$ so that 
queries are supported in $O(1)$ time and updates are supported in 
$O(\log^{\eps'})$ randomized time. Using $E_{ij}(v)$, it is straightforward 
to report all $p\in L_{ij}(v)$ with $p.y\in [c,d]$ in $O(K/B)$ I/Os. 
Consider a query $Q=[a,b]\times [c,d]$. 
We can find in $O(\log_B N)$ I/O operations 
$O(\log_B N)$ nodes $u^t$ and ranges $[i_t,j_t]$, 
so that the $x$-coordinate of a point $p$ belongs to $[a,b]$ if and 
only if $p$ is stored in some list $L_{i_tj_t}(u^t)$.  
Hence, all points in $Q$ can be reported by reporting all points in 
$L_{i_tj_t}(u^t)$ whose $y$-coordinates belong to $[c,d]$.
The total query cost is $O(\log N/\log\log N)=O(\log_B N)$. 
However, the space usage is $O((N/B)\log_2^{1+8\eps'}N)$  because each point is 
stored in $O(B^{2\eps'}\log_B N)=O(\log_2^{1+8\eps'}N)$ lists $L_{ij}(v)$. 
We can reduce the  space usage if  only parts of lists $L_{ij}(v)$  
 stored explicitly. 

Let $L(v)$ denote the list of all points that belong to a node $v$ sorted 
by their $y$-coordinates. We divide $L(v)$ into groups of points 
$G_s(v)$, $s=O(|L(v)|/B^{1+2\eps'})$, so that each 
$G_s(v)$ contains at least $B^{1+2\eps'}/2$ and at most 
$2B^{1+2\eps'}$ points.  
%improvement: increase group size to  $B^{1+2\eps'}\log_B N$
% the same for $\tL_{ij}(v)$
Instead of $L_{ij}(v)$, we store the list $\tL_{ij}(v)$. 
The main idea of our space saving method is that we need  to store 
points of $G_s(v)$ in the list $\tL_{ij}(v)$ 
only in the case when $G_s(v)$ contains a few points from $L_{ij}(v)$.  
Otherwise all relevant points can be found by querying the set 
$G_s(v)$ provided that $\tL_{ij}(v)$ contains a pointer to $G_s(v)$.
Points and pointers are stored in each list $\tL_{ij}(v)$ according to the 
following rules. 
If $|L_{ij}(v)\cap G_s(v)| \leq B/2$, the list $\tL_{ij}(v)$ contains all points 
from $L_{ij}(v)\cap G_s(v)$.  If $L_{ij}(v)\cap G_s(v) \geq  2B$,  
the list $\tL_{ij}(v)$ contains a pointer $\ptr_s$ to $G_s(v)$. 
We also store the minimal and maximal $y$-coordinates of points in 
$L_{ij}(v)\cap G_s(v)$ with each pointer to $G_s(v)$ from $\tL_{ij}(v)$.
If $B/2 < |L_{ij}(v)\cap G_s(v)| < 2B$, $\tL_{ij}(v)$ contains either 
a pointer to $G_s(v)$ or all points   from $L_{ij}(v)\cap G_s(v)$.

Instead of $E_{ij}(v)$, we will use several other auxiliary data 
structures. 
A data structure $\tE_{ij}(v)$ contains information about elements of  
$\tL_{ij}(v)$.  For each  point $p\in \tL_{ij}(v)$ we store $p.y$ in 
$\tE_{ij}(v)$; for every pointer $\ptr_s$, $\tE_{ij}(v)$ contains 
both the minimal and the maximal $y$-coordinate associated with 
$\ptr_s$. 
A data structure $E(v)$ contains the $y$-coordinates of all points 
in $L(v)$. Both $E(v)$ and all $E_{ij}(v)$ support one-reporting 
queries as described above.
A data structure $H_s(v)$  supports orthogonal 
range reporting queries on $G_s(v)$.  Using the data structure described in 
Lemma 1 of~\cite{ASV99}, 
we can answer three-sided reporting queries in $O(K/B)$ I/Os 
using $O(|G_s(v)|/B)$ blocks of space.  Using the standard approach, we can 
extend this result to a data structure that uses $O((|G_s(v)|\log_2 B)/B)$ 
blocks and answers queries in $O(K/B)$ I/Os.  
%improvement: 
% we can answer queries in $O(\log_B N +K/B)$ I/Os and use 
%$O((|G_s(v)|/B)\log_2 B/\log_2\log_B N)$ space

Now we show how we can report all points $p\in L_{ij}(v)$ with $p.y\in [c,d]$ 
without storing $L_{ij}(v)$.  
We can find an element $e$ of $\tL_{ij}(v)$ with $y$-coordinate in $[c,d]$. 
Suppose that such $e$ is found.
Then, we traverse the list  $\tL_{ij}(v)$ in $+y$ direction starting at $e$ 
 until a point $p$ with $p.y> d$ 
or a pointer to $G_s(v)$ with the minimal $y$-coordinate larger than $d$ 
is found.  
We also traverse $\tL_{ij}(v)$ in $-y$ direction until   a point $p$ with 
$p.y< c$ or a pointer to $G_s(v)$ with the maximal $y$-coordinate smaller 
than $c$ is found.  For every pointer in the traversed portion of 
$\tL_{ij}(v)$, we visit the corresponding group  $G_s(v)$ and 
report all points  $p\in G_s(v)\cap L_{ij}(v)$ with 
$p.y\in [c,d]$. 
All relevant points in $G_s(v)$ can be reported 
in $O(K_s/B)$ I/Os using the data structure $H_s(v)$; here $K_s$ denotes 
the number of points reported by $H_s(v)$. 
  %$([a,b]\times [c,d])$. 
% For a point $p \in G_s(v)$, $p.x\in [a,b]$ if and only if $p$  belongs 
%to a child $v_r$ of $v$, $i\leq r\leq j$. 
By definition of $\tL_{ij}(v)$, a set $G_s(v)\cap L_{ij}(v)$ contains at 
least $B/2$ points  if 
there is a pointer $\ptr$ from $\tL_{ij}(v)$ to $G_s(v)$. 
Unless $\ptr$ is the first or the last element in the traversed portion 
of $\tL_{ij}(v)$, $G_s(v)$ contains $B$ points from $[a,b]\times [c,d]$. 
Since $B$ consecutive elements of the list $\tL_{ij}(v)$ contain 
either $B$ points or at least one pointer to a group $G_s(v)$, the total cost 
of reporting all points in $L_{ij}(v)$ with $p.y\in [c,d]$ is $O(1+K/B)$. 

Now we consider the situation when there is no  $e\in \tE_{îj}(v)$, 
such that $e\in [c,d]$. 
In this case $L_{ij}(v)$ may contain some points from the range $Q$ 
only if all points 
$p\in L(v)$ with $p.y\in [c,d]$ belong to one group $G_s(v)$. 
Using $E(v)$, we search for  a point $p_s\in L(v)$ such that $p_s.y\in [c,d]$. 
If there is no such $p_s$, then $L(v)\cap Q =\emptyset$. 
Otherwise $p_s\in G_s(v)$ and 
we can report all points in $Q\cap G_s(v)$ in $O(1+K/B)$ I/Os 
using $H_s(v)$. 
We need to visit $O(\log_B N)$ nodes of the 
range tree to answer the query; hence, the total query cost is 
$O(\log_B N +K/B)$ I/Os.

Since the lists $L_{ij}(v)$ are not stored, the space usage is reduced to 
$O((N/B)\log_2 N)$: Each list $\tL_{ij}(v)$ contains less than $B$ points 
and at most one pointer for each group $G_s(v)$. Since $L(v)$ is divided  
into $O(|L(v)|/B^{1+2\eps'})$ groups, the total size of 
all $L_{ij}(v)$ is $O(|L(v)|)$. All data structures $H_s(v)$ for all groups 
$G_s(v)$ use $O((|L(v)|\log_2 B)/B)$ blocks of space. Each 
point belongs to  $O(\log_B N)$ nodes; therefore the total space usage 
is $O((N/B)\log_2 N)$. 

When a new point $p$ is inserted, we must insert it into $O(\log_B N)$ lists 
$L(v)$. Suppose that $p$ is inserted into $G_s(v)$ in a node $v$.
We insert $p$ into $H_s(v)$ in $O(\log_2B)$ I/Os;
$p$ is also inserted into up to $B^{2\eps'}=O(\log_2^{8\eps'}N)$ 
lists $\tL_{ij}(v)$. 
The one-dimensional reporting 
data structure for  $\tL_{ij}(v)$ supports updates in 
$O(\log_2^{\eps'}N)$ I/Os; hence, the total cost of inserting 
a point is $O(\log_2^{9\eps'}N)$. 
For each pair $i\leq j$, we check whether the number of points in 
$L_{ij}(v)\cap G_s(v)$ equals  to $2B$. 
Although the list $L_{ij}(v)$ is not stored, we can estimate the number 
of points in $L_{ij}(v)\cap G_s(v)$ by a query to the data structure 
$H_s(v)$. If $|L_{ij}(v)\cap G_s(v)|= 2B$,  we remove all points 
of $\tL_{ij}(v)\cap G_s(v)$ from $L_{ij}(v)$ and insert a pointer to 
$G_s(v)$ into $\tL_{ij}(v)$. Points in a list 
$\tL_{ij}(v)$ are replaced with a pointer to a group $G_s(v)$ at most 
once for a sequence of $\Theta(B)$ insertions into $G_s(v)$. Hence,
 the amortized 
cost of updating $\tL_{ij}(v)$  because the number of points from 
$L_{ij}(v)$ in a group exceeds $2B$ is $O(\log_2^{\eps'}N)$ I/Os.
Each insertion affects $O(\log_2^{8\eps'}N)$ lists $L_{ij}(v)$. 
 If the number of points in $G_s(v)$ 
equals $2B\log_2^{1+2\eps'}N$, we split the group  $G_s(v)$ into 
$G_1(v)$ and $G_2(v)$ of $B\log_2^{1+2\eps'}N$ points each. 
Since up to $B$ elements can be inserted and deleted into every 
list $\tL_{ij}(v)$, the amortized cost incurred by splitting a group 
is $O(\log_2^{9\eps'}N)$. 
Thus the total cost of inserting a point into data structures 
associated with a node $v$ is $O(\log_2^{9\eps'}N)$ I/Os. 
Since a new point is inserted into $O(\log_2 N/\log_2\log_2N)$ nodes 
of the range tree, the total cost of an insertion is 
$O(\log_2^{1+9\eps'}N/\log_2B)=O(\log_B^{1+\eps}N)$.  
Deletions are processed in a symmetric way. 

Combining this result with Lemma~\ref{lemma:2dim}, we obtain the following 
Theorem
\begin{theorem}
\label{theor:2dim3}
Suppose that  point coordinates are integers. There exists a data 
structure that 
uses $O((N/B)\log_2 N)$ blocks of space and answers orthogonal range 
reporting queries in two dimensions 
in $O(\log_B N+ \frac{K}{B})$ I/O operations. 
Updates are supported in $O(\log_B^{1+\eps} N)$ amortized I/Os w.h.p. for any 
$\eps>0$.
\end{theorem}

\section*{Appendix A. Two-Dimensional Range Reporting for
 $B=\Omega(\log_2^4 N)$ }
We maintain a constant fan-out tree $T$ on the set of $x$-coordinates 
of all points. An internal node of $T$ has at most eight children. 
 A point $p$ \emph{belongs} to an internal 
node $v$, if its $x$-coordinate is stored in a leaf descendant of $v$.
We assume that the height of $T$ is bounded by $h_1\log_2 N$.  
Each node $v$ contains two secondary data structures that support 
 three-sided queries of the 
form $\halfrightsect{a}{+\infty}\times [c,d]$ and 
$\halfleftsect{-\infty}{b}\times [c,d]$ respectively; both 
data structures contain all points that belong to $v$.  
We also store all points that belong to $v$ in a B-tree sorted by
their  $y$-coordinates, 
so that all points with $y$-coordinates in an interval $[c,d]$ can 
be reported. The data structures for three-sided queries 
are implemented as described in Lemma~\ref{lemma:3sid}. 
We implement the B-tree using the result of~\cite{A03}, so that updates 
are supported in $O(1/B^{\delta})$ I/Os. 
Hence, updates on the secondary data structures are supported in 
$O(1/B^{\delta})$ I/Os.

We say that a node $v$ of $T$ is \emph{special} if the depth of $v$ is 
divisible by $\ceil{\delta\log_2 B/3}$. To facilitate the query processing, 
buffers with inserted and deleted elements will be stored 
 in the special nodes only. 
A node $u$ is a \emph{direct special descendant} of $v$ if $u$ is a special 
node, $u$ is a descendant of $v$, and there is no other special node $u'$ on the
path from $v$ to $u$. We denote by $\desc(v)$ the set of direct special descendants of a node $v$. The set of nodes $\subs(v)$ consists of the node $v$ 
and all nodes $w$, such that $w$ is a descendant of $v$ and $w$ is an 
ancestor of some node $u\in desc(v)$. In other words, every node $w$ 
on a path from $v$ to one of its direct special descendants belongs to 
$\subs(v)$; the node $v$ also belongs to $\subs(v)$.

Let $\oI(v)$ and $\oD(v)$ denote the buffers of inserted and deleted 
points stored in a node $v\in T$. When a point is inserted, we add 
it to the buffer $\oI(v_R)$, where $v_R$ is the root of $T$.  
When a buffer $I(v)$ contains at least $B^{2\delta}$ elements, we visit 
every node $w\in \subs(v)$ and insert all points $p\in I(v)\cap rng(w)$ 
into the secondary data structures of a node $w$. 
Then, we examine all nodes $u\in \desc(v)$. For every $u\in \desc(v)$, 
we insert all points $p\in \oI(v)\cap rng(u)$ into 
$\oI(u)$ and remove all points $p\in (\oI(v)\cap rng(u))\cap \oD(u)$ 
from $\oD(u)$.  Finally, we set $I(v)=\emptyset$. 

The total number of nodes in $\subs(v)$ and $\desc(v)$ is $O(B^{\delta})$. 
% We can insert $B^c$ points into each data structure stored in a node 
%$w\in\subs(v)$ in $O(1)$ I/Os amortized. 
Since each point is inserted into  $O(\log_2 B)$ data structures 
and the total number of points is $O(B^{2\delta})$, 
all data structures in $\subs(v)$ can be updated in 
$O(B^{2\delta}\log_2 B/B^{\delta}) =O(B^{\delta}\log_2 B)$ I/Os. 
We can also update the buffers 
$\oI(u)$ and  $\oD(u)$ for each $u\in \desc(v)$ in $O(1)$ I/Os. 
Hence, a buffer $I(v)$ can be emptied in $O(B^{\delta}\log_2 B)$ I/Os. 
Since a buffer $I(v)$ is emptied once after $\Theta(B^{2\delta})$ points were 
inserted into $I(v)$, the amortized cost of an insertion into $I(v)$ 
is $O(\log_2 B/B^{\delta})$. 
An insertion of a point $p$ into the data structure 
leads to insertions of $p$ into $O(\log_B N)$  buffers $I(v)$.
Hence, the amortized cost of inserting a point $p$ is 
$O(\log_2 N/B^{\delta})=O(1)$.
Deletions can be processed with a simmetric procedure.

Consider a query $Q=[a,b]\times [c,d]$. We identify the node $v$ of $T$ 
such that $[a,b]\subset rng(v)$, but $[a,b]\not\subset rng(v_i)$ for any 
child $v_i$ of $v$. Suppose that $[a,b]$ intersects with $rng(v_l)$, $\ldots$,
$rng(v_r)$ where $1\leq l\leq r\leq 4$. All points $p\in S\cap Q$ are stored 
in the secondary structures of nodes $v_l,\ldots, v_r$ or in 
buffers of the special ancestors of $v$ (possibly 
including the node $v$ itself).  
We start by constructing sets   $INS(v)$ and $DEL(v)$. 
The set $INS(v)$ contains all points $p$ such that $p\in \oI(w)$ for an 
ancestor $u$ of $v$, but $p\not\in \oD(u')$ for an ancestor $u'$ of $u$. 
The set $DEL(v)$ contains all points $p$ such that $p\in \oD(w)$ for an 
ancestor $u$ of $v$, but $p\not\in \oI(u')$ for an ancestor $u'$ of $u$.
Only $O(\log_B N)$ ancestors of $v$ are special nodes and every buffer 
stored in a special node contains at most $B^{2\delta}$ points.    
Hence, both $INS(v)$ and $DEL(v)$ can be constructed in $O(\log_B N)$ I/Os 
and contain $h_1\cdot B^{2\delta}\log_B N \leq B/4$ points.
We output all points of $p\in INS(v)\cap Q$ in $O(1)$ I/Os.  
Let $\cV$ be the list of all points $p$, such that $p$ belongs to 
$Q$ and $p$ is stored in a child of $v$. The list $\cV$ can be generated as 
follows.
First, we answer three-sided queries $\halfrightsect{a}{+\infty} \times [c,d]$
and $\halfleftsect{-\infty}{b}\times [c,d]$ on data structures for nodes 
$v_l$ and $v_r$ respectively. Then, we identify all points $p$ stored in 
a node $v_j$, $l<j<r$, such that $c\leq p.y\leq d$. 
When the list $\cV$ is constructed, we traverse $\cV$ and output all points 
of $\cV$ that do not belong to the set $DEL$. 
The list $\cV$ can be generated and traversed in 
$O(\log_B N +\frac{|\cV|}{B})$ I/Os. Since the total number of points in the
answer is $K\geq |\cV|-B/4$, all points of $\cV\setminus DEL$ can be identified 
and reported in $O(\log_B N +\frac{|K|}{B})$ I/Os.

Our result is summed up in the following lemma
\begin{lemma}\label{lemma:2dim}
Suppose that $B^{\delta}\geq 4\log_2 N$ for a constant $\delta\leq 1/4$.
Then there exists a data structure that uses $O((N/B)\log_2 N)$ blocks of space 
and answers two-dimensional orthogonal range  reporting 
queries in $O(\log_B N + K/B)$ I/Os. 
The amortized cost of inserting or deleting a  point 
is $O(1)$.
\end{lemma}
We can slightly improve the space usage by increasing the fan-out of the 
base tree. Our construction is the same as above, but every internal 
node has $\Theta(\log_B N)$ children. We  also store an additional data 
structure $H(v)$ in every internal node $v$ of $T$. 
For any $l< r$ and any $c\leq d$, $H(v)$ enables us to efficiently 
 report all points $p$, 
such that $p.y\in [c,d]$ and $p$ is also stored in a child  $v_j$ of $v$ 
for $l<j <r$.
The data structure $H(v)$ is described below.
%We will show below that $H(v)$ supports 
%two-dimensional queries in $O(\log_B N+ K/B)$ I/Os and updates in 
%$O(1/B^{\delta})$ I/Os.

Let $L(v)$ denote the list of all points that belong to $v$.
Let $Y(v)$ be the set that contains  $y$-coordinates of all points in $L(v)$. 
For every point $p\in L(v_l)$ and for all children $v_l$ of $v$, 
$H(v)$ contains a ``point'' $\tau(p)=(p.y, succ(p.y,Y(v_l)))$. 
For a query $c$, $H(v)$ returns  all points $p\in L(v)$ such that
$\tau(p)\in \halfleftsect{-\infty}{c} \times \halfrightsect{c}{+\infty}$. 
In other words, we can report all points $p\in L(v)$ such that
$p.y\leq c$ and   $succ(p.y,Y(v_l)\geq c$. An answer to query 
contains $O(\log_B N)$ points; at most one point for each child $v_l$. 
Using Lemma~\ref{lemma:3sid},
$H(v)$ supports queries and updates in $O(\log_B N +K)=O(\log_B N)$ I/Os
 and updates in $O(1/B^{\delta})$ amortized I/Os respectively.  
 
We can report all points $p\in L(v_j)$ such that $p.y\in [c,d]$ and 
$l<j< r$ as follows. Using $H(v)$, we search for  all points $p$, such 
that $p.y\leq c$ and $succ(p.y,Y(v_l)\geq c$. For every found $p$, 
$l<j<r$, we traverse the list $L(v_j)$ and report all points 
that follow $p$ until a point $p'$, $p'.y>d$, is found. 
The total query cost is $O(\log_B N + K/B)$. 

% We find all elements in $H(v)$ that belong to 
% the range $\halfleftsect{-\infty}{c}\times \halfrightsect{c}{+\infty}$. 
% The answer contains $O(\log_B N)$ elements. For every element $e$
% in the answer, we find the corresponding point $p_e$. If $p_e\in L(v_j)$
% and $l<j<r$, we traverse the list $L(v_j)$ and report all points 
% that follow $p_e$ until a point $p'$, $p'.y>d$, is found. 

The global data structure  supports insertions and deletions of points  in the same way 
as shown in Lemma~\ref{lemma:2dim}.  The query answering procedure is 
also very similar to the procedure in the proof of Lemma~\ref{lemma:2dim}. 
We identify the node $v$ of $T$ 
such that $[a,b]\subset rng(v)$, but $[a,b]\not\subset rng(v_i)$ for any 
child $v_i$ of $v$. We also find the children $v_l,\ldots, v_r$ of $v$ such 
that $[a,b]$ intersects with $rng(v_l)$, $\ldots$, $rng(v_r)$. 
The  sets $INS(v)$, $DEL(v)$, and the list $\cV$  can be generated as 
described above.  
The only difference is that we identify all points $p$ stored in 
nodes $v_j$, $l<j<r$, such that $c\leq p.y\leq d$ using the data structure 
$H(v)$.

\end{document}